\documentclass[12pt]{article}
\usepackage{amsmath,amsfonts, amsthm,amssymb,hyperref,verbatim}
\newtheorem{prop}{Proposition}
\newtheorem{lemma}{Lemma}
\newtheorem{thm}{Theorem}
\newtheorem*{rmk}{Remark}
\newtheorem*{rmks}{Remarks}
\newtheorem{cor}{Corollary}

\begin{document}

\title{Bounds on the Lyapunov exponent via crude estimates on the density
of states}
\author{Mira Shamis\textsuperscript{1} and Thomas Spencer\textsuperscript{2}}
\maketitle
\begin{abstract}  We study the Chirikov (standard) map
at large coupling $\lambda \gg 1$, and prove that the Lyapounov exponent
of the associated Schr\"odinger operator is of order
$\log \lambda$ except for a set of energies of measure
$\exp(-c \lambda^\beta)$ for some $1 < \beta < 2$.
We also prove a similar (sharp) lower bound on the Lyapunov exponent
(outside a small exceptional set of energies) for a large family
of ergodic Schr\"odin\-ger operators, the prime example
being the $d$-dimensional skew shift.
\end{abstract}
\footnotetext[1]{Department of Mathematics, Princeton University, Princeton, NJ 08544, USA, and School of Mathematics, Institute for Advanced Study,
Einstein Dr., Princeton, NJ 08540, USA. E-mail: mshamis@princeton.edu. Supported in part by NSF grants PHY-1104596 and DMS-1128155.}
\footnotetext[2]{School of Mathematics, Institute for Advanced Study,
Einstein Dr., Princeton, NJ 08540, USA. E-mail: spencer@ias.edu.}

\section{Introduction}

\subsection{Standard map}

The Chirikov or Standard map is the two-dimensional area preserving dynamical system $T$ of the torus $\mathbb{T}^2 = (\mathbb{R} / 2\pi \mathbb{Z})^2$ to itself defined by
\[ T(x_1, x_2) = (x_2, 2x_2 + \lambda \sin x_2 - x_1)~.\] 
For small $\lambda$ the Chirikov map is known to have a set of invariant curves on which the motion is quasi-periodic. This follows from theory of Kolmogorov, Arnold and Moser (KAM). In fact by a theorem of Duarte \cite{D1,D2}, elliptic  islands are known to appear for an open dense set of $\lambda$. 

One of the major problems in dynamical systems is to prove that there is a set of initial conditions of \textit{positive measure} for which the dynamics is chaotic, 
i.e.\ the Lyapunov exponent is positive.  Equivalently, the  Kolmogorov--Sinai metric entropy $h(T)$ is conjectured to be positive, and to be
of order $\ln \lambda$ when $\lambda \gg 1$. The conjecture remains unproved for any value of $\lambda$. 

\vspace{1mm}\noindent
The orbit of the Chirikov map is determined by its initial condition $(x_{-1}, x_0)$ and may be expressed in the form
 
\[ \cdots, (x_{-1}, x_0), (x_0, x_1), (x_1, x_2), \cdots~, \]
where $x_j$ satisfy the equation for the discrete time pendulum
\[ (\triangle x)_n=  x_{n+1} + x_{n-1}  - 2 x_n = \lambda \sin x_n~.\]
Setting 
\[ \psi(n) = \frac{\partial x_n}{\partial x_0}~,\]
we obtain an equation for the linearization about an orbit:
\[  H_\omega \psi = \left(\frac{1}{\lambda} \triangle  - \cos x_n \right)\psi = 0 .\]
Here $H_\omega$ is the discrete Schr\"odinger operator associated with the map $T$. The potential $v=- \cos x_n$ is
evaluated along the orbit above and will depend on the initial condition $\omega = (x_{-1}, x_0)$.
By Pesin's entropy formula \cite{HP},
the metric entropy of $T$ is equal to the integral over $(x_{-1}, x_0)$ of the Lyapunov exponent of $H$ at energy 0 (which corresponds to energy $E=2\lambda^{-1}$ after the diagonal part of the Laplacian is incorporated in the potential).
The precise definition of the Lyapunov exponent (in the non-ergodic setting) is given in
Section~\ref{s:prelim} below. We can also define the average Lyapunov exponent $\gamma(E)$ at energy $E$ by studying
the related equation $ (H - E) \psi = 0$. 

Our main theorem roughly states that $\gamma (E) \approx  \log \lambda$
except for a set of $E$ of Lebesgue measure less than  $ e^{-c\lambda^{\beta}}$ for $1< \beta <2$.   
 
\begin{thm}\label{thm:stdmap} Fix $\epsilon > 0$. For any $\beta < 4/3$ and  sufficiently large $\lambda$,
\[ \mathrm{meas}\, \left\{ -1/2 \leq E \leq 1/2 \, \mid \, \gamma(E) < (1 - \epsilon) \ln \lambda \right\} 
	< C_{\epsilon} \exp\left[- C_{\epsilon}^{-1} \lambda^{\beta} \right]~. \]
For any $\beta < 2$, and any $E_0 \in (-1, 1)$, there exists a set $\Lambda(E_0)$ of $\lambda$ so that
for every $\lambda \in \Lambda(E_0)$
\[ \mathrm{meas}\, \left\{ E_0 - \lambda^{-2} \leq E \leq E_0 + \lambda^{-2} \, \mid \, 
	\gamma(E) < (1-\epsilon) \ln \lambda  \right\} 
	< C_{\beta,\epsilon} \exp\left[- C_{\beta,\epsilon}^{-1} \lambda^{\beta}\right]~. \]
and 
\[ \mathrm{meas} \left\{ \Lambda(E_0) \cap [2\pi \ell, 2\pi (\ell+1)] \right\} \to 2\pi \quad 
\text{as} \quad \ell \to \infty~. \]
\end{thm}
Here and further $\mathrm{meas}$ denotes Lebesgue
measure.

\begin{rmks}\hfill
\begin{enumerate}
\item We are mostly interested in energy $E_0 = 0$, since the metric entropy is equal to the average Lyapunov exponent at $E = 0$. 
\item The interval $(-1/2, 1/2)$ in the first estimate can be extended to $(-1 + \xi, 1-\xi)$, for an arbitrary
$\xi > 0$.
\item The set $\Lambda(E_0)$ may be chosen to be 
the complement of the Minkowski (element-wise) sum
\[ A + \bigcup_{\ell = 1}^\infty (2\pi \ell -\ell^{-\delta}, 2 \pi \ell + \ell^{-\delta}) \] 
for some finite set $A \subset (-\pi,\pi]$ depending on $E_0$ and some small $\delta > 0$ (see the proof of Lemma~\ref{l:bound.K}). For $E_0 = 0$, one
may choose $A = \{0, \pi\}$.
\item Instead of eliminating a set of energies of small measure we could have added a random potential (or noise) to the Schr\"odinger operator with variance 
$\approx e^{ - \lambda ^{\beta}}$.  
\item For any $E$, one has the complementary inequality $\gamma(E) \leq \ln \lambda + C$; see Section~\ref{s:prelim}. This remark also applies to the
setting of Theorem~\ref{th2} below.
\end{enumerate}
\end{rmks}

\vspace{2mm}\noindent
Our proof relies on a formula of Jones and Thouless relating the density of states $\rho(E)$ to the Lyapunov exponent
$\gamma(E)$. To get a lower bound on $\gamma(0)$ it is sufficient get some mild H\"older regularity of $\rho$ near $E=0$ for large $\lambda$. This
observation, in  more general context,
goes back to the work of Avron, Craig, and 
Simon \cite{ACS}.
Although H\"older regularity of the density
of states is not known for the standard
map, we obtain our theorem by bounding $\rho (-\delta, \delta)$ for  small $\delta$ depending on $\lambda$.

In \cite{Sp}, the second-named author proved an inequality similar to that of Theorem~\ref{thm:stdmap} with the bound $\approx e^{- C^{-1} \frac{\lambda}{\ln \lambda}}$
on the size of the exceptional set. While our argument is
based on the strategy of \cite{Sp}, the latter only uses the regularity of the distribution of $f$, whereas
our result requires an analysis of resonances. 

There is also partial  unpublished work of Carelson and the second-named author \cite{CaSp} which tried to extend
the  ideas of Benedicks and Carleson  \cite{BC} on the H\'enon map to the standard map.  See also the paper by 
Ledrappier et al.\ \cite {LSSW} for other approaches to metric entropy. 

In \cite{B}, Bourgain showed that, for sufficiently large $\lambda$, $\gamma(E)$ is positive outside
a set of energies of zero measure. While Bourgain's result has a much smaller exceptional set
than Theorem~\ref{thm:stdmap}, his method does not yield an explicit lower bound on the Lyapunov
exponent.

A recent paper by Gorodetski \cite{Gor} proves
that there is a set  of Hausdorff dimension $2$ for which the Lyapunov exponent is positive for all $\lambda>0$. 
This paper also gives an up to date overview of  the dynamics of the standard map.

\subsection{General bounds on the Lyapunov exponent}

The second result of this paper pertains to a large family of Schr\"odinger operators. The setting is
as follows. $(\Omega , \mathcal{B},\mu)$ is an arbitrary probability space, $T: \Omega \to \Omega$ is
a measure-preserving map, and the one-dimensional Schr\"odinger operators $H_\omega$, $\omega \in \Omega$,
are constructed via
\[ (H\psi)(n) = \lambda^{-1} (\psi(n-1)+\psi(n+1)) + V(n)\psi(n)~,\]
where
\[ V(n) = V_\omega(n) = f(T^n\omega) \]
for some bounded measurable $f:\Omega \to \mathbb{R}$. The definition of the Lyapunov exponent associated
with $H$ and $E \in \mathbb{R}$ is discussed in Section~\ref{s:prelim}.

\begin{thm}\label{th2}
Let $\ell \geq 1$ be a natural number, and assume that $2\ell-1$  consequtive values $V(-\ell+1),V(-\ell+2),\cdots,V(\ell-1)$
of the potential are independent bounded random variables, and that the restriction of their common distribution
to an interval $[a, b]$ has bounded density. Then, for any closed interval $I \subset (a, b)$ and large $\lambda \gg 1$,
\[ \left| \left\{ E \in I \, \mid \, \gamma(E) \leq \ln \lambda - C \ln \ln \lambda \right\} \right| 
  \leq C \exp \left[ - \lambda^\ell \frac{\ln \ln \lambda}{(C \ln \lambda)^{2\ell-1}} \right]~.\]
\end{thm}

For $\ell = 1$, this result can be derived from the arguments in \cite{Sp}.

\paragraph{Example} Let $\Omega = \mathbb{T}^d = (\mathbb{R}/(2\pi\mathbb{Z}))^d$ be the $d$-dimensional
torus equipped with the normalized Lebesgue measure. Fix $\alpha \notin \mathbb{Q}$, and consider the skew shift
\[ T(\omega_1, \cdots, \omega_d) = (\omega_1 + \alpha, \omega_2 + \omega_1, \cdots, \omega_d + \omega_{d-1})~. \]
Let $h: \mathbb{T} \to \mathbb{R}$ be a function such that
\begin{enumerate}
 \item $h \in C^2(\mathbb{T})$
 \item $h$ has one non-degenerate maximum $M$ and one non-degenerate minimum $m$ (without loss of generality
$M = h(0) = h(2\pi)$, $m = h(c)$ for some $0 < c < 2\pi$)
 \item $h$ is strictly increasing on $[0, c]$ and strictly decreasing on $[c, 2\pi]$
 \item $|h''(c)|, |h''(0)|>0$.
\end{enumerate}
Let $f(\omega) = h(\omega_d)$ and let $H = H_\omega$ be the corresponding Schr\"odinger operator. Set $\ell = \lfloor \frac{d-1}{2} \rfloor$.

\begin{cor}
For any $\xi > 0$,
\[ \mathrm{meas} \, \left\{ m+\xi \leq E \leq M - \xi \, \big| \,
  \gamma(E) \leq \ln \lambda - C \ln \ln \lambda \right\} 
\leq C_\xi \exp \left[ -  \lambda^\ell \frac{\ln \ln \lambda}{(C\ln\lambda)^{2\ell-1}}\right] \]
\end{cor}

The corollary can be extended to more general functions $h$. For example, if $h$ is a Morse function,
the bound holds in any interval away from the critical values of $h$. 

We remark that a related result for the skew shift was obtained by Chan, Goldstein, and Schlag \cite{CGS}.
There, the assumptions on the function $h$ are weaker, and the lower bounds on the Lyapunov exponent are complemented
by a proof of Anderson localization; on the other hand, the exceptional set of energies of \cite{CGS} is much larger 
than ours (a power of $\lambda^{-1}$).

Stronger results have been obtained for the skew shift
in the case when the function $h$ is analytic; see the book of Bourgain \cite{Bourg}.

\section{Preliminaries}\label{s:prelim}

Let $(\Omega, \mathcal{B}, \mu)$ be a probability space, and let $T: \Omega \to \Omega$ be a
measure-preserving invertible transformation. We shall consider one-dimensional discrete
Schr\"odinger operators $H = H_\omega$, $\omega \in \Omega$, acting on 
$\ell^2(\mathbb{Z})$ by
\[ (H\psi)(n) = \lambda^{-1} (\psi(n-1) + \psi(n+1)) + V(n)\psi(n)~,\]
where $\lambda \gg 1$ is the coupling constant, and $V = V_\omega$ is constructed
from $T$ and a bounded measurable function $f: \Omega \to \mathbb{R}$ by
the formula
\[ V_\omega(n) = f(T^n \omega)~. \]

If the transformation $T$ is ergodic, the operator $H$ is called ergodic. In this case
the Lyapunov exponent $\gamma(E; \omega)$ is defined for every 
$E \in \mathbb{R}$ and $\mu$-almost every $\omega$; moreover, $\gamma(E; \omega)$
is equal to an $\omega$-independent number $\gamma(E)$ on a set of $\omega$ of
full measure. In addition, the density of states $\rho$ is defined (as an $\omega$-independent
probability measure on $\mathbb{R}$). The Lyapunov exponent is related to the density
of states by the Thouless formula
\begin{equation}\label{eq:thj}
\gamma(E) = \ln \lambda + \int \ln |E - E'| d\rho(E')~. 
\end{equation}
These facts may be found for example in the book \cite{CFKS} of Cycon, Froese, Kirsch,
and Simon.

If $T$ is not ergodic, the Lyapunov exponent and the density of states admit a following
generalization, based on the theorem on ergodic decomposition (first proved by
von Neumann \cite{vN} and Krylov--Bogolyubov \cite{KB}). Let us briefly recall the definitions.

A $T$-invariant probability measure $\eta$ on $(\Omega, \mathcal{B})$
is called $T$-ergodic if $T$ is ergodic on $(\Omega, \mathcal{B}, \eta)$. The space of
all $T$-ergodic measures is denoted $\mathrm{Erg}(T)$. The ergodic decomposition theorem
(see Walters \cite[pp.~27--28]{W}) states that there exists a probability measure $\eta$ on $\mathrm{Erg}(T)$ so that
\begin{equation}\label{eq:decomp}
\mu = \int_{\mathrm{Erg} (T)} \nu \, d\eta(\nu)~.
\end{equation}
The representation (\ref{eq:decomp}) is called the ergodic decomposition of $\mu$.

\vspace{2mm}\noindent
By Fubini's theorem, one can prove the following:
\begin{lemma} The Lyapunov exponent $\gamma(E; \omega)$is defined for every 
$E \in\mathbb{R}$ and $\mu$-almost every $\omega \in \Omega$. The density of
states $\rho_\omega$ is defined for $\mu$-almost every $\omega$. 
\end{lemma}
\begin{proof}
Let $B$ be the set of $\omega$ for which $\gamma(E; \omega)$ is not defined. Then,
for any $\eta \in \mathrm{Erg}\,T$, $\eta(B) = 0$ by the ergodic case. Hence by
(\ref{eq:decomp}), $\mu(B) = 0$. The second statement is proved in a similar way.
\end{proof}

\noindent We set
\begin{equation}\label{eq:defgamma}
\gamma(E) = \int \gamma(E, \omega) d\mu(\omega)
\end{equation}
and
\begin{equation}\label{eq:defrho}
\rho = \int \rho_\omega d\mu(\omega)~.
\end{equation}
Fubini's theorem also yields the following lemmata:
\begin{lemma}\label{l:bdrho}
For every $z = E_0 + i\delta \in \mathbb{C} \setminus \mathbb{R}$,
\[ \int \frac{d\rho(E)}{E-z} = \int (H_\omega - z)^{-1}(0,0) d\mu(\omega)~.\]
In particular, for any $0 < \alpha \leq 1$,
\begin{multline*}
 \rho(E_0 - \delta, E_0 +\delta) 
\leq (2\delta)^\alpha \int \left[ \Im (H_\omega - E_0 - i\delta)^{-1}(0,0) \right]^\alpha d\mu(\omega) \\
  \leq (2\delta)^\alpha \int |(H_\omega - E_0 - i\delta)^{-1}(0,0)|^\alpha d\mu(\omega)~.
\end{multline*}
 
\end{lemma}

\begin{lemma}\label{l:thj} The Thouless-Jones fornula (\ref{eq:thj}) remains valid with the 
definitions (\ref{eq:defgamma}), (\ref{eq:defrho}).
\end{lemma}
Lemma~\ref{l:thj} yields the following upper bound on the Lyapunov exponent (which can be also easily obtained by other means):
\[\begin{split}
\gamma(E) &\leq \ln \lambda + \max_{E' \in \sigma(H)} \ln|E - E'| \\
  &\leq \ln \lambda + \ln \mathrm{diam}\,\sigma(H)\\
  &\leq \ln \lambda + \ln(\max f - \min f + 4 \lambda^{-1}) \leq \ln \lambda + C~. 
\end{split}\]

\section{Application of the Thouless--Jones formula}

Fix an energy $E_0 \in \mathbb{R}$, $t > 0$, and $1 \geq \xi > \delta > 0$. Denote
\[ g(\delta) = g(\delta; E_0) = \max\left(\delta, \, \sup_{|E - E_0| \leq \xi} \rho[E-\delta, E+\delta]\right) \]
and
\[ Z_t = Z_t(E_0) = \left\{ E \, \Big{|} \, |E - E_0| \leq \delta \quad \text{and} \quad \gamma(E) \leq t \right\}~. \]
\begin{prop}\label{prop:thj}
In the notation above, for any $\delta > 0$ we have:
\[ \mathrm{meas}\, Z_t(E_0) \leq 2e \exp \left\{ - \frac{\ln \lambda - t - 6 \xi \ln \frac{e^2 g(\delta; E_0)}{\xi \delta}}
						              {2g(\delta; E_0)}\right\}~.\]
\end{prop}
\begin{rmk}
The main goal of this paper is to make  the right-hand side as small
as possible by getting good estimates on $\rho(E-\delta, E+\delta)$
for small $\delta$.
\end{rmk}

\begin{proof}[Proof of Proposition~\ref{prop:thj}] From the Thouless--Jones formula (see (\ref{eq:thj}) and Lemma~\ref{l:thj}),
\[ \gamma(E) = \ln \lambda + \int \ln |E - E'| d\rho(E')~. \]
Therefore by definition of $Z_t = Z_t(E_0)$
\[\begin{split}
t \, \mathrm{meas}\, Z_t
	&\geq \int_{Z_t} \gamma(E) \, dE \\
	&= \int_{Z_t} \left\{ \ln \lambda + \int \ln|E-E'| d\rho(E')\right\} dE \\ 
	&\geq \mathrm{meas}\, Z_t\, \ln \lambda 
		- \int_{Z_t} dE \int_{|E' - E| \leq 1} \ln|E-E'|^{-1} d\rho(E')~.
\end{split}\]
Hence
\begin{equation}\label{eq:beforecases}\begin{split}
(\ln \lambda - t) \,\mathrm{meas}\,Z_t
	&\leq \int_{Z_t} dE \int_{|E'-E|\leq 1} \ln |E-E'|^{-1} d\rho(E')\\
	&= \int_{E_0-1-\delta}^{E_0+1+\delta} d\rho(E') 
		\int_{Z_t \cap \{|E-E'|\leq1\}} \ln|E-E'|^{-1} dE~. 
\end{split}\end{equation}
Denote
\[ J(E') = \int_{Z_t \cap \{|E-E'|\leq1\}} \ln|E-E'|^{-1} dE \]
and decompose
\begin{multline*} \int_{E_0-1-\delta}^{E_0+1+\delta} J(E') d\rho(E')\\
= \left\{ \int_{|E'-E_0| \leq 2\delta} +
\sum_{k=1}^{k_0-1} \int_{2k\delta \leq |E' - E_0| \leq (2k+2)\delta} + \int_{|E'-E_0|\geq 2k_0 \delta}\right\} J(E') d\rho(E')~, 
\end{multline*}
where we shall later  choose $k_0 = \lfloor \frac{\xi}{2g(\delta)}\rfloor+1$.

Consider the following cases:\\
\vspace{1mm}\noindent{\bf i)}
$|E' - E_0| \leq 2\delta$; then by the rearrangement inequality
\[\begin{split}
J(E') &\leq \int_{Z_t} \ln |E-E'|^{-1} dE \\
	&\leq \int_{E' - \mathrm{meas} Z_t / 2}^{E' + \mathrm{meas} Z_t/2} \ln |E-E'|^{-1} dE \\
	&= \mathrm{meas} Z_t \ln \frac{2e}{\mathrm{meas} Z_t}~.
\end{split}\]
Therefore
\[ \int_{|E'-E_0| \leq 2\delta} J(E') d\rho(E') \leq 
	2 g(\delta) \mathrm{meas} Z_t \ln \frac{2e}{\mathrm{meas} Z_t}~. \]

\vspace{1mm}\noindent
{\bf ii)} $2k\delta \leq |E' - E_0| < (2k+2) \delta$ for some $1 \leq k \leq \xi \delta^{-1} -2$.
Then for $E \in Z_t$ we have $|E - E_0| \leq \delta$ and hence $|E-E'| \geq (2k-1)\delta$.
Therefore
\[ \int_{2k\delta \leq |E'-E_0| \leq (2k+2)\delta} J(E') d\rho(E')
	\leq 2g(\delta) \mathrm{meas} \, Z_t \, \ln \frac{1}{(2k-1)\delta}~.\]

\vspace{1mm}\noindent
{\bf iii)} By the same reasoning, 
\[ \int_{2k_0 \delta \leq |E' - E_0|} J(E') d\rho(E') 
	\leq \mathrm{meas}\,Z_t \, \ln \frac{1}{(2k_0-1)\delta}~.  \]

\vspace{1mm}\noindent Combining these estimates, we obtain:
\[\begin{split}
&\int_{E_0-1-\delta}^{E_0+1+\delta} J(E') d\rho(E')  \\
&\quad\leq \mathrm{meas} \, Z_t \, \left\{
	2 g(\delta) \ln \frac{2e}{\mathrm{meas} \, Z_t}
	+ 2 g(\delta) k_0 \ln \frac{e^2}{(2k_0-1)\delta}
	+ \ln \frac{1}{(2k_0-1)\delta} \right\}~.
\end{split}\]
Taking $k_0 = \lfloor \frac{\xi}{2g(\delta)} \rfloor + 1$ and plugging into
(\ref{eq:beforecases}) yields:
\[ \ln \lambda - t \leq 2 g(\delta) \ln \frac{2e}{\mathrm{meas} Z_t} 
	+ 6\xi \ln \frac{e^2 g(\delta)}{\xi \delta}~,\]
whence 
\[ \mathrm{meas} Z_t \leq 2e 
	\exp \left\{ - \frac{\ln \lambda - t - 6\xi \ln \frac{e^2 g(\delta)}{\xi \delta}}
                               {2g(\delta)}\right\}~.  \]
\end{proof}

\section{Theorem~\ref{th2}: proof}

Let $H_{m} = H_{m,\omega}$ be the restriction of $H = H_\omega$ to $\{0,\cdots,m-1\}$. Denote
\[ \Delta_{m}(\omega) = \det (H_{m,\omega} - E - i\delta)~. \] 
For $m = 0$, set $\Delta_0(\omega) = 1$. The proof of Theorem~\ref{th2} uses the following auxiliary proposition,
which we prove after we prove the theorem.
\begin{prop}\label{prop:2} Let $m \geq 1$ be a fixed integer.
Suppose $V(0),\cdots,V(m-1)$ are independent, and that their common distribution has bounded
density $\leq A$ in $[E-\xi,E+\xi]$. Then for any $a$ with $\Im a \geq 0$ and $|a| \leq \xi/2$,
\[ \int d\omega |\Delta_m(\omega) - a \Delta_{m-1}(\omega)|^{-1} \leq (3 A \ln (1 + 1/(A\delta)) + 2\xi^{-1})^m~. \]
\end{prop}
\begin{rmk}
Only the case $a=0$ is needed to prove the theorem; however, the stronger statement is more suited
for an inductive proof and may have other applications.
\end{rmk}

\begin{proof}[Proof of Theorem~\ref{th2}]
Let $G_\omega(z) = (H_\omega - z)^{-1}$, and $G_{\omega,{2\ell-1}} = (H_{\omega,2\ell-1} - z)^{-1}$, where
$H_{\omega,2\ell-1}$ is the restriction of $H_\omega$ to $\{0,1,\cdots,2\ell-2\}$. According to the resolvent
identity,
\begin{multline}\label{eq:resolv}G_\omega(E+i\delta; \ell-1,\ell-1) = G_{\omega,2\ell-1}(E+i\delta; \ell-1,\ell-1)  \\
 + \frac{1}{\lambda} \Big[ G_{\omega,2\ell-1}(E+i\delta; \ell-1,0) G_\omega(E+i\delta; -1, \ell-1) \\
 +  G_{\omega,2\ell-1}(E+i\delta; \ell-1,2\ell-2) G_\omega(E+i\delta; 2\ell-1, \ell-1) \Big]~. \end{multline}
Let $M = \|V-E-i\delta\|_\infty$. By Cramer's rule we can express the matrix elements of $G_{\omega,2\ell-1}$ as ratios of two determinants. Using the  inequality
\begin{equation}
|\Delta_{m}(\omega)| \leq \left(M+\frac{1}{\lambda}\right)^{m}~, \quad m = 0,1,2,\cdots,
\end{equation}
 (cf.\ (\ref{eq:rec}) below), and recalling that our matrix is of size $(2\ell-1)\times(2\ell-1)$ with indices numbered from $0$ to $2\ell - 2$, we obtain:
\[  |G_{\omega,2\ell-1}(E+i\delta; \ell-1,\ell-1) | \leq \left(M+\frac{1}{\lambda}\right)^{2(\ell - 1)} |\Delta_{2\ell - 1}(\omega)|^{-1}~,\]
whereas
\[ |G_{\omega,2\ell-1}(E+i\delta; \ell-1,0)| \leq \lambda^{-(\ell-1)}  \left(M+\frac{1}{\lambda}\right)^{\ell - 1} |\Delta_{2\ell - 1}(\omega)|^{-1}\]
and
\[ |G_{\omega,2\ell-1}(E+i\delta; \ell-1,2\ell-2)| \leq \lambda^{-(\ell-1)}  \left(M+\frac{1}{\lambda}\right)^{\ell - 1} |\Delta_{2\ell - 1}(\omega)|^{-1}~. \]
Combining these inequalities with the bound $\| G_\omega(E+i\delta)\| \leq \delta^{-1}$,
we obtain from (\ref{eq:resolv}) with $\delta = \lambda^{-\ell}$ and $\lambda \geq 2$:
\[\begin{split} |G_\omega(E+i\delta; \ell-1,\ell-1)|
&\leq \left[ (M+\frac{1}{\lambda})^{2(\ell-1)} + \lambda^{-\ell} \frac{2}{\delta} (M + \frac{1}{\lambda})^{\ell-1}\right]
|\Delta_{2\ell - 1}(\omega)|^{-1} \\
 &\leq (M+3)^{2(\ell-1)} |\Delta_{2\ell - 1}(\omega)|^{-1}~.\end{split}\]
By Proposition~\ref{prop:2}
\[\begin{split}
 &\int d\mu(\omega) |G_\omega(E+i\delta; \ell-1,\ell-1)|\\
  &\qquad\leq (M+3)^{2(\ell-1)} (3 A \ln (1 + 1/(A\delta)) + 2\xi^{-1})^{2\ell-1} \\
  &\qquad \leq (C \ln \lambda)^{2\ell - 1}~,
  \end{split}\]
where $A$ is a bound on the density of $V$ in $[a,b]$, and $\xi$ is the distance from $I$ to $\{a,b\}$. 
By Proposition~\ref{prop:thj}, one can choose $\widetilde{C}$ so that
\[ \mathrm{meas} \, Z_{\ln \lambda - C \ln \ln \lambda}(E) \leq 2 e^{- \lambda^\ell \frac{\ln \ln \lambda}{(C \ln \lambda)^{2\ell-1}}}~. \]
\end{proof}

The proof of Proposition~\ref{prop:2} is based on two  estimates (the first one follows from a rearrangement 
inequality, the second one follows from the first one).

\begin{lemma}\label{l:bdddens}
Let $\tau$ be a sub-probability measure (i.e.\ $\tau(\mathbb{R}) \leq 1$) with density bounded by $A$. Then
\[ \int \frac{d\tau(v)}{\sqrt{|v-E|^2+\delta^2}} \leq 3 A \ln(1 + 1/(A\delta))\]
for any $E \in \mathbb{R}$ and $\delta > 0$.
\end{lemma}

\begin{lemma}\label{l:2}
Let $\sigma$ be a probability measure on $\mathbb{R}$ so that the restriction of $\sigma$ to
$[-\xi, \xi]$ has density $\leq A$. Then for any real $a$ with $|a| \leq |\xi|/2$, 
\[ \int \frac{d\sigma(v)}{\sqrt{(v-a)^2 + \delta^2}} \leq 3 A \ln(1+1/(A\delta)) + 2\xi^{-1}~.\]
\end{lemma}
\begin{proof}[Proof of Proposition~\ref{prop:2}]
Without loss of generality we may set $E = 0$. Denote
\[ F_m (a) = \int \frac{d\sigma(V(0))d\sigma(V(1)) \cdots d\sigma(V(m-1))}{|\Delta_m - a \Delta_{m-1}|}~, \quad M_m = \max F_m~. \]
The proof is by induction on $m$. For $m=1$ the statement follows directly from Lemma~\ref{l:2}. For the induction
step, note that
\begin{equation}\label{eq:rec}\Delta_m = V(m-1) \Delta_{m-1} - \lambda^{-2} \Delta_{m-2}~.\end{equation}
Represent $\sigma = \sigma_1 + \sigma_2$, where $\sigma_1 = \sigma|_{[-\xi, \xi]}$; then
\[ \mathrm{supp} \sigma_1 \subset [-\xi, \xi]~, \quad \mathrm{supp} \sigma_2 \cap (-\xi, \xi) = \varnothing~,\]
and the density of $\sigma_1$ is bounded by $A$. Then
\[\begin{split}
F_m(a) &= \int \frac{d\sigma(V(0))\cdots d\sigma(V(m-1))}{|\Delta_m - a \Delta_{m-1}|}  \\
&= \left[\int d\sigma(V(0))\cdots d\sigma_1(V(m-1)) + \int d\sigma(V(0))\cdots d\sigma_2(V(m-1))\right]\\
&\qquad\qquad \frac{1}{|\Delta_m - a \Delta_{m-1}|} \\
&= I_1 + I_2~.
 \end{split}\]
The first integral is equal to
\[ I_1 = \int \frac{d\sigma(V(0)) \cdots d\sigma(V(m-2))}{|\Delta_{m-1}|} \int \frac{d\sigma_1(V(m-1))}{|V(m-1) - a_1|}~, \]
where $a_1 = a + \lambda^{-2} \Delta_{m-2} \Delta_{m-1}^{-1}$ has non-negative imaginary part. Hence by Lemma~\ref{l:bdddens}
\begin{equation}\label{eq:I1}
I_1 \leq 3 A \ln (1 + 1/(A\delta)) F_{m-1}(0) \leq 3 A \ln(1+1/(A\delta)) M_{m-1}~. 
\end{equation}
On the other hand,
\[ I_2 = \int \frac{d\sigma_2(V(m-1))}{|V(m-1) - a|} \int \frac{d\sigma(V(0)) \cdots d\sigma(V(m-2))}{|\Delta_{m-1} - a_2 \Delta_{m-2}|}~. \]
where $a_2 = \lambda^{-2} (V(m-1) - a)$. Then $\Im a_2 \geq 0$ and
\[ |a_2| \leq \frac{\lambda^{-2}}{\xi/2} \leq \xi/2 \]
for sufficiently large $\lambda$. Therefore
\begin{equation}\label{eq:I2}
I_2 \leq \int  \frac{d\sigma_2(V(m-1))}{|V(m-1) - a|} F_{m-1}(a_2) \leq \frac{2}{\xi} M_{m-1}~. 
\end{equation}
Combining (\ref{eq:I1}) with (\ref{eq:I2}), we obtain
\[ M_m \leq (3 A \ln (1 + 1/(A\delta)) + 2\xi^{-1}) M_{m-1}~. \]
\end{proof}

\section{Standard map: proofs}\label{s:pfstd}

Let $z = E+i\delta$, and set
\[\begin{split}
\Delta_3(z; x_0, x_1) &= \Delta_3(z; \omega) = \det \left( \begin{array}{ccc}  
- \cos x_{-1} - z & \frac{1}{\lambda} & 0 \\
\frac{1}{\lambda} & - \cos x_0 - z & \frac{1}{\lambda}\\
0 & \frac{1}{\lambda} & - \cos x_1 - z \ \end{array} \right) \\
&= - \Big[ c_{-1} ( c_0 c_1 - \frac{1}{\lambda^2}) - \frac{1}{\lambda^2} c_1 \Big] \\
&= - c_0 \Big[  c_1 c_{-1} - \frac{1}{\lambda^2 c_0} \left( c_1 + c_{-1} \right) \Big]~,
\end{split}\]
where 
\[\begin{split}
c_1 &= \cos x_1 + z \\
c_0 &= \cos x_0 + z \\
c_{-1} &= \cos x_{-1} + z = \cos (2 x_0 + \lambda \sin x_0 - x_1) + z~.
\end{split}\]
Let $a>0$ be a large number independent of $\lambda$
and define 
\[ A_a = \left\{ (x_0, x_1) \, \mid \, |c_0(x_0, x_1)| \geq a \lambda^{-2} \right\} = A_a^0 \times [-\pi, \pi]~. \]
Let
\[ I = I_\lambda(z; a) = \int_{A_a} |\Delta_3(z; x_0, x_1)|^{-\alpha} dx_0 dx_1~. \]

The next lemma is the main result of this section.
\begin{lemma}\label{l:main.stdmap}
For sufficiently large $a$ and $\lambda$, the following estimates hold:
\begin{enumerate}
\item $I_\lambda(E, a) \leq C_{\alpha}$ for any $\alpha < 2/3$;
\item $I_\lambda(E, a) \leq C_{\alpha}$ if  $\lambda$ is outside a small exceptional set (in 
the sense of Theorem~\ref{thm:stdmap} and the subsequent remark) and $\alpha < 1$.
\end{enumerate}
Here $C_{\alpha}>0$ depends only on $\alpha$.
\end{lemma}

\begin{proof}[Proof of Theorem~\ref{thm:stdmap} using Lemma~\ref{l:main.stdmap}]
We apply Proposition~\ref{prop:thj} with $\delta = \lambda^{-2}$. The density 
of states is bounded as follows. Let $G_\omega(z) = (H_\omega - z)^{-1}$; to bound 
$G_\omega(z; 0, 0)$, we introduce the restriction $H_\omega^{(3)}$ of $H_\omega$
to $\{-1,0,1\}$, and denote its Green function by 
$G_\omega^{(3)}(z) = (H_\omega^{(3)} - z)^{-1}$. We have:
\[\begin{split}
G_\omega(z; 0, 0) = G_\omega^{(3)}(z;0,0) 
	&+ \lambda^{-1} G_\omega^{(3)}(z; 0,1) G_\omega(z;2,0)\\
	&+ \lambda^{-1} G_\omega^{(3)}(z; 0,-1) G_\omega(z;-2,0)~.
\end{split}\]
By Cramer's rule,
\[ |G_\omega^{(3)}(z; 0,0)|\leq C| \Delta_3(z; \omega)|^{-1}~;
\quad |G_\omega^{(3)}(z; 0, \pm1)| \leq C \lambda^{-1} | \Delta_3(z; \omega)|^{-1}~. \]
For $\Im z = \delta$, we also have the trivial bound $\|G_\omega\| \leq \delta^{-1} = \lambda^2$. Thus
\begin{multline*} |G_\omega(z; 0, 0)| \\ \leq |G_\omega^{(3)}(z;0,0)|
	+ \frac{1}{\lambda\delta} 
	(|G_\omega^{(3)}(z; 0,1)| + |G_\omega^{(3)}(z; 0,-1)|) \leq C | \Delta_3(z; \omega)|^{-1}~.
	\end{multline*}
 Therefore
\[ \begin{split}
    \int |G_\omega(z; 0,0)|^\alpha d\mu(\omega)
      &\leq C^\alpha \int_{A_a} |\Delta_3(z; \omega)|^{-\alpha} d\mu(\omega) + (1-\mu(A_a)) \lambda^2 \\
      &= 4\pi^2 C^\alpha I_\lambda(z; a) + (1 - \mu(A_a))\lambda^2~.
   \end{split}\]
It is easy to see that $\mu(A_a) \leq C_a \lambda^{-2}$ (since $E$ is away from $\pm 1$). Also,  $|\Delta_3(z; x_0, x_1)|$ is monotone in
$\Im z$, since $\Delta_3(\cdot; \omega)$ has real roots. By Lemma~\ref{l:main.stdmap},
\[ \int |G_\omega(z; 0,0)|^\alpha d\mu(\omega) \leq C' \]
for sufficiently large $a$, under the assumptions of Theorem~\ref{thm:stdmap}. By Lemma~\ref{l:bdrho}
\[ \rho(E - \delta, E+\delta) \leq C'' \delta^\alpha~.\]
By Proposition~\ref{prop:thj} we obtain the statement.
\end{proof}

\begin{proof}[Proof of Lemma~\ref{l:main.stdmap}]

First, we write $I$ as
\[\begin{split}
I &= \int_{A_a^0} \frac{dx_0}{|\cos x_0 + E|^\alpha} \int_{-\pi}^\pi 
	\frac{dx_1}{|c_1 c_{-1} - \frac{1}{\lambda^2 c_0} (c_1 + c_{-1})|^\alpha} \\
  &= \int_{A_a^0} \frac{dx_0}{|\cos x_0 + E|^\alpha} J(E, \theta(x_0), \alpha, \frac{1}{\lambda^2(\cos x_0 + E)})~,
\end{split}\]
where 
\[\theta(x_0) = 2 x_0 + \lambda \sin x_0~,\]
\[ g(x_1) = g(E,\epsilon,\theta; x_1) = (\cos x_1 + E)(\cos(x_1-\theta) + E) - \epsilon((\cos x_1 + E) + (\cos(x_1 - \theta) + E))~,\]
and 
\[ J(E,\theta,\alpha; \epsilon) = \int_{-\pi}^\pi 
	\frac{dx_1}{|g(E,\epsilon,\theta,x_1)|^{\alpha}}~. \]
\begin{lemma}\label{l:7}
There exists $\epsilon_0 > 0$ such that for any $\frac{1}{2} < \alpha < 1$ and $|\epsilon| \leq \epsilon_0$,
\[ J(E, \theta,\alpha,\epsilon) \leq C \mathrm{dist}^{-(2\alpha-1)}(\bar\theta, \{0, \pm 2\arccos (- E)\})~, \]
where $\bar\theta = \theta \mod 2\pi$ and the constant $C$ depends only on $\alpha$.
\end{lemma}
Note that $J$ does not depend on $\lambda$, hence neither does $\epsilon_0$. We will later apply this lemma with $\epsilon = (\lambda^2 (\cos x_0 + E))^{-1}$.

\begin{proof}Without loss of generality we may assume that $|\theta| \leq \pi$. As a function of $x$, $g(x)$ is 
a trigonometric polynomial of degree $2$; hence it has $4$ zeros. For $\epsilon = 0$,
the zeros are $\pm x^*, \pm x^* + \theta$, where $x^* = \arccos (-E)$. Since  $E \in  (-1/2, 1/2)$ is bounded away
from $\pm 1$, $x^*$ is bounded away from $-x^*$, and $x^* + \theta$ is bounded away from $-x^* + \theta$. Consider 
four cases.

\vspace{2mm}\noindent{\bf Case 1:} $\theta$ is bounded away from $0$, $2x^*$. Then the $4$ zeros of 
$g(E,0,\theta; \cdot)$ are separated; hence the same is true for $g(E,\epsilon,\theta;\cdot)$ when $\epsilon$
is sufficiently small, $|\epsilon| \leq \epsilon_0(\theta)$. Therefore
\[ J(E,\theta,\alpha,\epsilon) \leq C~.\]

\vspace{2mm}\noindent{\bf Case 2:} $\theta$ is close to $0$. Then $x^* + \theta$ is close to $x^*$
and $-x^* + \theta$ is close to $-x^*$. The corresponding zeros $x_{1/2}$ and $x_{3/4}$ of
$g(E,\epsilon,\theta;\cdot)$ satisfy 
\[ |x_1 - x_2| \geq C^{-1} |\theta|~, \quad |x_3 - x_4| \geq C^{-1}|\theta|~,\]
as one can see, for example, plugging the linear approximation
\[ \cos x + E \simeq - \sin x^* (x - x^*)~, \quad \cos(x - \theta) + E \simeq - \sin x^* (x - \theta - x^*)\]
near $x^*$ into the definition of $g$. Therefore
\[ J(E,\theta,\alpha,\epsilon) \leq C|\theta|^{-(2\alpha-1)}~,\]
where we have used the following observation: 
for any $d$ such that $0 < |d| < 1$ ($d$ may be complex) and 
$1/2 < \alpha < 1$,
\begin{equation}\label{l:integrquadr} \int |x(x-d)|^{-\alpha} dx \leq C_\alpha |d|^{-(2\alpha-1)}~. 
\end{equation}

\vspace{2mm}\noindent{\bf Cases 3 and 4}: $\theta$ is close to $\pm 2 x^*$. Then by similar reasoning
\[ J(E,\theta,\alpha,\epsilon) \leq C|\theta \mp 2 x^*|^{-(2\alpha-1)}~.\]
\end{proof}

\vspace{2mm}\noindent
To complete the proof of Lemma~\ref{l:main.stdmap}, we need to estimate
\[\begin{split}
&\int_{-\pi}^\pi \frac{dx_0}{|\cos x_0 + E|^\alpha} 
	\mathrm{dist}^{-(2\alpha-1)}(\bar{\theta}(x_0), \{0, \pm 2x^*\}) \\
&\quad\leq \int_{-\pi}^\pi \frac{dx_0}{|\cos x_0 + E|^\alpha} 
	\mathrm{dist}^{-(2\alpha-1)}(\bar{\theta}(x_0), 0)\\
&\quad\qquad+ \int_{-\pi}^\pi \frac{dx_0}{|\cos x_0 + E|^\alpha} 
	\mathrm{dist}^{-(2\alpha-1)}(\bar{\theta}(x_0), 2x^*) \\
&\quad\qquad+\int_{-\pi}^\pi \frac{dx_0}{|\cos x_0 + E|^\alpha} 
	\mathrm{dist}^{-(2\alpha-1)}(\bar{\theta}(x_0), -2x^*)~. 
\end{split}\]
Let 
\[ K(\lambda, b, E, \alpha)
	= \int_{-\pi}^\pi \frac{dx_0}{|\cos x_0 + E|^\alpha} 
	\mathrm{dist}^{-(2\alpha-1)}(\bar{\theta}(x_0), b)~. \]

\begin{lemma}\label{l:bound.K}
If $E$ is bounded away from $\pm 1$ and $\alpha < 2/3$,
\[ K(\lambda,b,E,\alpha) < C_\alpha~. \]
For  $\lambda$ outside a small exceptional set (satisfying the measure estimate in Theorem~\ref{thm:stdmap}), the 
same estimate holds for any $\alpha < 1$.
\end{lemma}

\begin{proof}
For simplicity of presentation, we first discuss in detail the case $E = 0$ and $b = 0$, and then comment on
the modifications needed for the general case. Setting
$x = \pi/2 - x_0$,
\begin{equation}\label{eq:k00}\begin{split}
K(\lambda, \alpha) &= K(\lambda,0,0,\alpha)
	= \int \frac{dx}{|\sin x|^\alpha}  
	\mathrm{dist}^{-(2\alpha-1)}(\bar{h}(x), 0)\\
    &= \int  \frac{dx}{|\sin x|^\alpha |\bar{h}(x)|^{2\alpha-1}}   ~, \end{split}\end{equation}
where
\[ h(x) = \lambda \cos x + 2x~, \]
and as before $\bar\bullet$ denotes reduction modulo $2\pi$. Now we need to estimate the integral over $[-\pi/2, \pi/2]$ and over $[-\pi, \pi] \setminus [-\pi/2,\pi/2]$. We show the argument for the former, since the argument for the latter is identical.

Decompose
\[ [-\pi/2,\pi/2] = \bigcup_{|\ell| \leq \left\lceil \frac{\lambda + \pi}{2\pi}\right\rceil} \left\{ (2\ell-1)\pi \leq h(x) \leq (2\ell+1)\pi\right\}~.\]
All but one set in this decomposition are unions of two intervals which do not contain the origin ($0$) 
with a zero (root) of $\bar{h}$ is each. One is a single interval (a neighborhood of $0$) with two zeros
of $\bar{h}$.

First consider the intervals
\[ I_\ell^+ \cup I_\ell^- = \left\{ (2\ell-1)\pi \leq h(x) \leq (2\ell+1)\pi\right\}  \]
which do not contain the origin. If $\ell$ is such that $|\lambda - 2\pi \ell| \geq c\lambda$,
where $c > 0$ is a small constant, we have the following estimates for $x \in I_\ell^\pm$:
\[ |x| \geq c_1~,\quad|h'(x)| \geq c_1\lambda~,\]
which imply that
\[ \sum_{|\lambda - 2\pi \ell| \geq c\lambda} \int_{I_\ell^+ \cup I_\ell^-} \frac{dx}{|\sin x|^\alpha |\bar{h}(x)|^{2\alpha - 1}}\]
is bounded for any $\alpha < 1$.

For the intervals with $|\lambda - 2\pi \ell| < c\lambda$ which do not contain the origin, we use the Taylor expansion
\begin{equation}\label{eq:tay} h(x) = \lambda ( 1 - \frac{x^2}{2} + O(x^4) ) + 2x~, \end{equation}
which implies that the zeros $x_\ell^\pm \in I_\ell^\pm$ of $\bar{h}$ are given by
\begin{equation}\label{eq:zeroasymp}
x_\ell^\pm = \frac2\lambda \pm \sqrt{\frac4{\lambda^2} 
	+ 2\,\frac{\lambda-2\pi\ell}\lambda + \epsilon_\ell^\pm}~, \quad
	|\epsilon_\ell^\pm| \leq C |x_\ell^\pm|^4~. 
\end{equation}
This relation yields the estimates
\[ |x| \geq \frac{c_2 |\lambda - 2\pi \ell|^{1/2}}{\lambda^{1/2}}~, \quad
|h'(x)| \geq c_2 |\lambda - 2\pi \ell|^{1/2}\lambda^{1/2} \]
for $x \in I_\ell^\pm$, as well as the bound 
\[ \mathrm{meas}\, I_\ell^\pm \leq 1 / (c_2 |\lambda - 2\pi \ell|^{1/2}\lambda^{1/2})~.\]
Hence 
\[ \int_{I_\ell^\pm} \frac{dx}{|\sin x|^\alpha |\bar{h}(x)|^{2\alpha - 1}} \leq \frac{C}{|\lambda - 2\pi \ell|^{\frac{1+\alpha}{2}} \lambda^{\frac{1-\alpha}{2}}}~. \]
Since $2\pi \ell$ is bounded away from $\lambda$ for intervals not containing the origin,
the sum of all these integrals is also bounded for any $\alpha < 1$. (This estimate is valid for any large $\lambda$, and the properties of $\lambda$ play a role only in the estimate for the interval with two zeros, which
we consider next.)

Finally, consider the interval $I_{\ell_0}$ containing two zeros $x_{\ell_0}^\pm$ of $\bar{h}$ (observe
that $\lambda - 2\pi \ell_0 = \bar{\lambda}$). From the Lagrange interpolation formula
\[ h(x) = h(x_{\ell_0}^+) \frac{x - x_{\ell_0}^-}{x_{\ell_0}^+-x_{\ell_0}^-} +  
h(x_{\ell_0}^-) \frac{x - x_{\ell_0}^+}{x_{\ell_0}^--x_{\ell_0}^+} + \frac{1}{2} h''(\tilde{x}) (x - x_{\ell_0}^+)(x - x_{\ell_0}^-)\]
($\tilde{x}$ lies in the interval spanned by $x$ and $x_{\ell_0}^\pm$) we have:
\begin{equation}\label{eq:lagr} |\bar{h}(x)| \geq \frac{1}{3} \lambda |x - x_{\ell_0}^+| |x - x_{\ell_0}^-|\end{equation}
for $x \in I_{\ell_0}$. Also the relation (\ref{eq:zeroasymp}) is still valid. If $|\bar{\lambda}| \geq \lambda^{-\delta}$, we have:
\[ |x_{\ell_0}^\pm| \geq c_3 / \lambda^{\frac{1+\delta}{2}}~,\]
whence the integral
\begin{equation}\label{eq:intl0}\int_{I_{\ell_0}} \frac{dx}{|\sin x|^\alpha |\bar{h}(x)|^{2\alpha - 1}} \end{equation}
is bounded for any $\alpha < 1$, provided that $\delta = \delta(\alpha)>0$ is sufficiently small. 

This establishes Theorem 1 with $\beta <2$  for $E=0$ for $\lambda \in \Lambda  = \left\{ |\bar\lambda| \ge \lambda^{-\delta} \right\}$.   Note  that $\bar\lambda\approx 0$  corresponds to the presence of an elliptic island. 

\vspace{2mm}\noindent
Without
restrictions on $\lambda$, (\ref{eq:zeroasymp}) implies that one has either 
\[ |x_{\ell_0}^+| \geq c/\lambda \quad \text{and} \text  |x_{\ell_0}^-| \geq c / \lambda\]
or
\[ |x_{\ell_0}^+ -x_{\ell_0}^-| \geq c/\lambda~.\]
In both cases, the integral (\ref{eq:intl0}) is bounded for $\alpha < 2/3$. For example, in the first case
(\ref{eq:lagr}) and the Cauchy--Schwarz inequality imply
\[ \begin{split}\int_{I_{\ell_0}} \frac{dx}{|\sin x|^\alpha |\bar{h}(x)|^{2\alpha - 1}}
&\leq \frac{C}{\lambda^{2\alpha-1}} \int_{I_{\ell_0}} \frac{dx}{|x|^{\alpha} |x-x_{\ell_0}^+|^{2\alpha-1}
|x-x_{\ell_0}^-|^{2\alpha-1}}\\
&\leq \frac{C}{\lambda^{2\alpha-1}} \int_{I_{\ell_0}} \frac{dx}{|x|^{\alpha} |x-c/\lambda|^{4\alpha-2}}~;
\end{split}\]
for $\alpha \leq 2/3$, the right-hand side does not
exceed
\[  \frac{C}{\lambda^{2\alpha-1}} \int_{I_{\ell_0}} \frac{dx}{|x|^{\alpha} |x-c/\lambda|^{\alpha}}~,\]
and this expression is bounded by 
a constant
due to
the estimate (\ref{l:integrquadr}).
This concludes  the proof of the lemma for the case $E = b = 0$. 

\vspace{2mm}\noindent
In the general case, set $\widetilde{x}^* = \pi - x^*$;
then (for $E$ away from $\pm 1$)
\[ |x_0 + E| \geq \frac{1}{C} \min(|x_0 - \widetilde{x}^*|, |x_0 + \widetilde{x}^*|)~, \]
hence we may replace the term $\cos x_0 + E$
in the denominator with $|x_0 \pm  \widetilde{x}^*|$.
Letting $x = x_0 \pm  \widetilde{x}^*$, we obtain
an integral similar to (\ref{eq:k00}), with the function
\[ 2 (x_0 \mp  \widetilde{x}^*) + \lambda \sin (x_0 \mp  \widetilde{x}^*)\]
in place of $h$.
The form of the Taylor expansion (\ref{eq:tay}) and its corollary
(\ref{eq:zeroasymp}) is similar to that in the case $b = E = 0$. The final step of the argument requires
no modification.
\end{proof}
\renewcommand{\qedsymbol}{}\end{proof}

\textbf{Acknowledgements:} We wish to thank J.~Bourgain, L.~Carleson, and Ya.~Sinai  for many helpful discussions.

\end{document}